\documentclass[10pt]{article}
\usepackage{amsbsy}
\usepackage{amsfonts}
\usepackage{amsmath}
\usepackage{amsthm}
\usepackage{amssymb}
\usepackage{apacite}
\usepackage{setspace}
\usepackage{graphicx,psfrag,epsf}
\usepackage{bm,multicol}
\usepackage[english]{babel}
\usepackage{natbib}
\usepackage[T1]{fontenc}
\usepackage[utf8]{inputenc}
\usepackage{authblk}
\usepackage{xcolor}
\usepackage{natbib}

\usepackage[title]{appendix}

\newtheorem{theorem}{Theorem}

\newtheorem{remark}{Remark}

\doublespace
\newcommand{\newc}{\newcommand}
\newc{\N}{\mbox{N}}
\newc{\1}{\bf{1}}
\def\signed #1{{\leavevmode\unskip\nobreak\hfil\penalty50\hskip2em
  \hbox{}\nobreak\hfil(#1)%
  \parfillskip=0pt \finalhyphendemerits=0 \endgraf}}

\newsavebox\mybox

\begin{document}
\title{Increasing the Replicability for Linear Models via Adaptive Significance Levels}
\author[1]{D. Vélez\footnote{daiver.velez@upr.edu, ORCID=0000-0001-7162-8848}}
\author[2]{ M.E. Pérez\footnote{maria.perez34@upr.edu, ORCID=0000-0001-8641-8405}}
\author[2]{L. R. Pericchi\footnote{luis.pericchi@upr.edu, ORCID=0000-0002-7096-3596}}
\affil[1]{\small University of Puerto Rico, Río Piedras Campus, Statistical Institute and Computerized Information Systems, Faculty of Business Administration, 15 AVE Universidad STE 1501, San Juan, PR 00925-2535, USA}
\affil[2]{University of Puerto Rico, Río Piedras Campus, Department of Mathematics, Faculty of Natural Sciences, 17 AVE Universidad STE 1701, San Juan, PR 00925-2537, USA}
\date{}
\maketitle
\begin{abstract}
We put forward an adaptive  alpha (Type I Error) that decreases as the information grows, for hypothesis tests in which nested linear models are compared. A less elaborate adaptation was already presented in \citet{PP2014} for comparing general i.i.d. models. In this article we present refined versions to compare nested linear models. This calibration may be interpreted as a Bayes-non-Bayes compromise, of a simple translations of a Bayes Factor on frequentist terms that leads to statistical consistency, and most importantly, it is a step towards statistics that promotes replicable scientific findings.

\end{abstract}


\newpage
\section{Motivation}
Obtaining a $p$-value lower than $0.05$ no longer opens the doors for publication, but now statisticians must provide alternatives to scientists. One of the most important problems in statistics and in science as a whole, is to provide statistical measures of evidence that lead to  replicable scientific findings. In this article, we propose an adaptive alpha level for linear models that depends on the design matrices, the difference in dimension between the models and the ``effective'' sample size. This adaptive alpha mimics the behaviour of a powerful Bayes Factor based on recent improvements of BIC (Bayarri et. al 2019)  but uses the familiar concepts of \textit{Significance Hypothesis Testing}.     

\section{Basic Derivation}
Consider the linear regression model $\mathbf{y}=\mathbf{X}\boldsymbol{\beta}+\boldsymbol{\epsilon}$, where $\mathbf{y}$ represents the n-dimensional random vector of response variables, $\mathbf{X}$ is the $n\times k$ matrix of non-stochastic explanatory variables (for simplicity, here we assume that $\mathbf{X}$ is a full rank matrix), $\boldsymbol{\beta}$ is a k-dimensional vector of regression parameters, $\boldsymbol{\epsilon}$ is an n-dimensional vector of standard normal errors, i.e. $\boldsymbol{\epsilon}\sim N(0,\sigma^2\mathbf{I}_n)$, and $\sigma$ is the standard deviation of the error, $\sigma>0$.

We denote with $M$ the full model whose matrix form is given by:

\begin{eqnarray*}
\mathbf{y}~~~~~&=&~~~~~~~~~~~~~~~\mathbf{X}~~~~~~~~~~~~~~~~~~~\boldsymbol{\beta}~~~~+~~~~~\boldsymbol{\epsilon}\\
\begin{bmatrix}
~y_1&\\
~y_2&\\
~\vdots\\
~y_n&\\
\end{bmatrix}&=&\begin{bmatrix}
1&x_{12}&x_{13}&\cdots & x_{1k}\\
1&x_{22}&x_{23}&\cdots & x_{2k}\\
\vdots &\vdots &\vdots & \ddots &\vdots\\
1&x_{n2}&x_{n3}&\cdots & x_{nk}\\
\end{bmatrix}\begin{bmatrix}
~\beta_1&\\
~\beta_2&\\
~\vdots\\
~\beta_k&\\
\end{bmatrix}+\begin{bmatrix}
~\epsilon_1&\\
~\epsilon_2&\\
~\vdots\\
~\epsilon_n&\\
\end{bmatrix}
\end{eqnarray*}

\noindent where $\epsilon_i\sim N(0,\sigma^2), \text{with}\hspace{.2cm} 1\leq i\leq n$.\\

Now suppose that we want to perform pairwise model comparisons between nested generic sub-models $M_i$ and $M_j$ from $M$, where $M_j$ is a sub-model having $j(\leq k)$ regression coefficients, with $M_i$ nested to $M_j$. Formally, we want to test the hypothesis
 $$H_i:\text{Model}\hspace{.1cm} M_i\hspace{0.3cm} versus \hspace{.3cm}H_j:\text{Model}\hspace{.1cm} M_j,$$ 
  
\noindent in other words, we are comparing the following two nested linear models $$M_i:\mathbf{y}=\mathbf{X}_i\boldsymbol{\delta}_i+\boldsymbol{\epsilon}_i,\hspace{.5cm}\boldsymbol{\epsilon}_i\sim N(0,\sigma_i^2\mathbf{I}_n)$$ and  $$M_j:\mathbf{y}=\mathbf{X}_j\boldsymbol{\beta}_j+\boldsymbol{\epsilon}_j,\hspace{.5cm}\boldsymbol{\epsilon}_j\sim N(0,\sigma_j^2\mathbf{I}_n).$$

So the Bayes Factor is:
$$B_{ij}(\mathbf{y})=\dfrac{\displaystyle\int f(\mathbf{y}|\mathbf{X}_i\boldsymbol{\delta}_i,\sigma_i^2\mathbf{I}_n)\pi^N(\boldsymbol{\delta}_i,\sigma_i)d\boldsymbol{\delta}_id\sigma_i}{\displaystyle\int f(\mathbf{y}|\mathbf{X}_j\boldsymbol{\beta}_j,\sigma_j^2\mathbf{I}_n)\pi^N(\boldsymbol{\beta}_j,\sigma_j)d\boldsymbol{\beta}_jd\sigma_j}.$$
The construction of the adaptive alpha is based on $B_{ij}(\mathbf{y})$, without explicit assessment of prior distributions by the user. Instead we will use well established statistical practices to directly construct summaries of evidence. 

\begin{itemize}
  \item[1.] \textit{Approximation to Bayes factors under regularity conditions:} Laplace's asymptotic method, under regularity conditions, gives the following approximation \citep[see for example][]{BerPer01}:
\begin{equation}\label{eq1}
B_{ij}=\dfrac{f(\mathbf{y}|\mathbf{X}_i\widehat{\boldsymbol{\delta}}_i,S_i^2\mathbf{I}_n)|\hat{I}_i|^{-1/2}}{f(\mathbf{y}|\mathbf{X}_j\widehat{\boldsymbol{\beta}}_j,S_j^2\mathbf{I}_n)|\hat{I}_j|^{-1/2}}\cdot\dfrac{(2\pi)^{i/2}\pi^N(\widehat{\boldsymbol{\delta}}_i,S_i)}{(2\pi)^{j/2}\pi^N(\widehat{\boldsymbol{\beta}}_j,S_j)},
\end{equation} 
  
\noindent where $\widehat{\boldsymbol{\delta}}_i, S_i^2, \widehat{\boldsymbol{\beta}}_j, S_j^2,$  are MLE's at the parameters and $\hat{I}_i, \hat{I}_j$ are the observed information matrices respectively for $M_i$ and $M_j$.
Since the first factor typically goes to $\infty$ or to $0$ as the sample size accumulates, but the second factor stays bounded, it is useful to rewrite (\ref{eq1}) as:

\begin{equation}\label{eq2}
-2\log(B_{ij})=-2\log\left(\dfrac{f(\mathbf{y}|\mathbf{X}_i\widehat{\boldsymbol{\delta}}_i,S_i^2\mathbf{I}_n)}{f(\mathbf{y}|\mathbf{X}_j\widehat{\boldsymbol{\beta}}_j,S_j^2\mathbf{I}_n)}\right)-2\log\left(\dfrac{|\hat{I}_j|^{1/2}}{|\hat{I}_i|^{1/2}}\right)+C.  
\end{equation}
 \item[2.]\textit{Likelihood ratio:} The likelihood ratio can be written as:
\begin{equation}\label{eq3}
    \dfrac{f(\mathbf{y}|\mathbf{X}_i\widehat{\boldsymbol{\delta}}_i,S_i^2\mathbf{I}_n)}{f(\mathbf{y}|\mathbf{X}_j\widehat{\boldsymbol{\beta}}_j,S_j^2\mathbf{I}_n)}=\left(\frac{S_{j}^2}{S_{i}^2}\right)^{\frac{n}{2}}=\left(\dfrac{\mathbf{y}^t(\mathbf{I}-\mathbf{H}_j)\mathbf{y}}{\mathbf{y}^t(\mathbf{I}-\mathbf{H}_i)\mathbf{y}}\right)^{\frac{n}{2}}
    \end{equation} 
See Appendix 1 for derivations.

 \item[3.]\textit{The Fisher information matrix:} The Observed Fisher Information Matrix (OFIM) with $i$ adjustable parameters is 

\begin{equation}\label{eq4}
\hat{I}_i(\widehat{\boldsymbol{\delta}_i})=\dfrac{1}{S_i^2}\cdot \mathbf{X}_i^t\mathbf{X}_i.
\end{equation}
Returning to equation (\ref{eq2}) and using (\ref{eq3}) and (\ref{eq4}) we have 

\begin{equation}\label{eq5}
-2\log(B_{ij})=-(n-1)\log\left(\dfrac{\mathbf{y}^t(\mathbf{I}-\mathbf{H}_j)\mathbf{y}}{\mathbf{y}^t(\mathbf{I}-\mathbf{H}_i)\mathbf{y}}\right)-\log\left(\dfrac{|\mathbf{X}_j^t\mathbf{X}_j|}{|\mathbf{X}_i^t\mathbf{X}_i|}\right)+C.
\end{equation}
 The constant $C$ depends on the prior assumptions and does not go to zero, but it is of lesser importance as the sample size grows.
\end{itemize}

\subsection{Sampling distribution of the likelihood ratio}
 Under $H_0$, the sampling distribution of  $\dfrac{\mathbf{y}^t(\mathbf{I}-\mathbf{H}_j)\mathbf{y}}{\mathbf{y}^t(\mathbf{I}-\mathbf{H}_i)\mathbf{y}}$ is a beta distribution, 
 \begin{equation}
    \dfrac{\mathbf{y}^t(\mathbf{I}-\mathbf{H}_j)\mathbf{y}}{\mathbf{y}^t(\mathbf{I}-\mathbf{H}_i)\mathbf{y}}\sim Beta\left(\dfrac{n-j}{2},\dfrac{q}{2}\right) 
 \end{equation}
 where $q=j-i$ \citep[see][Corollary 1]{art1}.

 \begin{theorem}
 \begin{equation}\label{eq7}
 -(n-1)\log\left(\dfrac{\mathbf{y}^t(\mathbf{I}-\mathbf{H}_j)\mathbf{y}}{\mathbf{y}^t(\mathbf{I}-\mathbf{H}_i)\mathbf{y}}\right)\sim Ga\left(\frac{q}{2},\frac{\frac{n-j}{n-1}}{2}\right)
 \end{equation}
 \end{theorem}
 \begin{proof}
 Let $Z=-(n-1)\log(Y)$ and $Y\sim Beta\left(\frac{n-j}{2},\frac{q}{2}\right)$, then
\begin{eqnarray*}
F_{Z}(z)=P(Z\leq z)&=&P(Y\geq e^{-\frac{z}{n-1}})\\
                   &=&1-F_Y(e^{-\frac{z}{n-1}})\\
                   & \Downarrow & \\
                  f_Z(z)&=&\left(\frac{1}{n-1}\right)e^{-\frac{z}{n-1}}f_Y(e^{-\frac{z}{n-1}}). \\
\end{eqnarray*}

\noindent Thus $$f_Z(z)=\frac{1}{n-1}\dfrac{\Gamma\left(\frac{n-j}{2}+\frac{q}{2}\right)}{\Gamma\left(\frac{n-j}{2}\right)\Gamma\left(\frac{q}{2}\right)}e^{-\left(\frac{n-j}{2(n-1)}\right)z}(1-e^{-\frac{z}{n-1}})^{\frac{q}{2}-1}$$
 but $\Gamma(n+\alpha)\approx \Gamma(n)n^{\alpha}$ \citep[see][eq. 6.1.46]{Abra}, so,
 
 \begin{eqnarray*}
f_Z(z)&=&\frac{1}{n-1}\dfrac{\left(\frac{n-j}{2}\right)^{\frac{q}{2}}}{\Gamma\left(\frac{q}{2}\right)}e^{-\left(\frac{n-j}{2(n-1)}\right)z}(1-e^{-\frac{z}{n-1}})^{\frac{q}{2}-1}\\
      &=& \dfrac{\left(\frac{n-j}{2(n-1)}\right)}{\Gamma\left(\frac{q}{2}\right)}e^{-\left(\frac{n-j}{2(n-1)}\right)z}\left(\frac{n-j}{2}-\frac{n-j}{2}e^{-\frac{z}{n-1}}\right)^{\frac{q}{2}-1}\\
      &=& \dfrac{\left(\frac{n-j}{2(n-1)}\right)}{\Gamma\left(\frac{q}{2}\right)}e^{-\left(\frac{n-j}{2(n-1)}\right)z}\left(\frac{n-j}{2(n-1)}z\right)^{\frac{q}{2}-1}+O(n^{-2})
\end{eqnarray*}

\noindent hence $Z\sim Ga\left(\frac{q}{2},\frac{\frac{n-j}{n-1}}{2}\right)$.
 \end{proof}
 
\begin{remark}
Theorem~1 is consistent with Wilk's Theorem, that is

\begin{center}
$Ga\left(\frac{q}{2},\frac{\frac{n-j}{n-1}}{2}\right)$ $\longrightarrow$ $\mathcal{X}^2_{(q)}$ as $n \rightarrow \infty,$
\end{center}

see for example \citet{CasBer2001}, Theorem 10.3.3.
\end{remark}

\subsection{Condition for the adaptive $\alpha$ to be approximately equivalent (yield the same decision) to a Bayes factor} 

If we denote by $g_{n,{\alpha}}(q)$ the quantile of the test statistic of (\ref{eq7}) corresponding to a tail probability $\alpha$, using (\ref{eq5}) and (\ref{eq7}) we can make an important departure from classical hypothesis testing: instead of fixing the tail probability (and the quantile) as in significance testing, we let the quantile vary according to the following rule

\begin{equation}\label{eq8}
g_{\alpha_{(\mathbf{X}_i,\mathbf{X}_j,n)}}(q)=g_{n,{\alpha}}(q)+\log\left(\dfrac{|\mathbf{X}_j^t\mathbf{X}_j|}{|\mathbf{X}_i^t\mathbf{X}_i|}\right).
\end{equation} 

\noindent Then the Bayes factor will converge to a constant (and $g_{\alpha_{(\mathbf{X}_i,\mathbf{X}_j,n)}}(q)$ will replace the fixed quantile). Note that (\ref{eq8}) establishes an approximate equivalence between Bayes Factor and adaptive significance levels.

\section{Adaptive alpha for linear models}
In order to establish the asymptotic correspondence between $\alpha$ levels and Bayes factor we need the following asymptotic expansion for the upper tail for large $Ga(\frac{q}{2},\frac{\frac{n-j}{n-1}}{2})=g_n(q)$,

\begin{equation}\label{eq9}
1-F(g_n(q))=1-Pr(g_n(q))\approx \dfrac{g_n(q)^{\frac{q}{2}-1}\exp\{-\frac{n-j}{2(n-1)}\cdot g_n(q)\}}{\left(\frac{2(n-1)}{n-j}\right)^{q/2-1}\Gamma\left(\frac{q}{2}\right)},
\end{equation}

\noindent see \citet{art2}.\\

Now we equate the significance level $\alpha$ to the approximate upper tail probability in (\ref{eq9}):

$$\alpha\approx \dfrac{g_{n,\alpha}(q)^{\frac{q}{2}-1}\exp\{-\frac{n-j}{2(n-1)}\cdot g_{n,\alpha}(q)\}}{\left(\frac{2(n-1)}{n-j}\right)^{q/2-1}\Gamma\left(\frac{q}{2}\right)}.$$

If we replace the fixed quantile $g_{n,\alpha}(q)$ by $g_{\alpha_{(\mathbf{X}_i,\mathbf{X}_j,n)}}(q)$ as in (\ref{eq8}), the following result is obtained: 

\begin{equation}\label{eq10}
\alpha_{(b,n)}(q)=\dfrac{[g_{n,\alpha}(q)+\log(b)]^{\frac{q}{2}-1}}{b^{\frac{n-j}{2(n-1)}}\cdot\left(\frac{2(n-1)}{n-j}\right)^{q/2-1}\Gamma\left(\frac{q}{2}\right)}\times C_{\alpha},
\end{equation}

\noindent where $b=\frac{|\mathbf{X}_j^t\mathbf{X}_j|}{|\mathbf{X}_i^t\mathbf{X}_i|}$. This is the simple (approximate) calibration we have been looking for, and defines the linear adaptive $\alpha_{(b,n)}$ levels and also the corresponding adaptive quantiles, which are suitable for constructing adaptive testing intervals for any $q$. Note that we still need to assign a value the constant $C_\alpha$ in (\ref{eq10}); this will be discussed in next section.

\begin{remark}\label{obs2}
Note that the adaptive significance level $\alpha_{(b,n)}$ depends on the design matrices, the sample size $n$ and the difference of dimension between the models being compared. In particular, large sample sizes and co-linearity among explanatory variables will affect the significance level.  
\end{remark}

\begin{remark}\label{obs3}
The derivation in this section will be further refined in next section along the lines of the Prior Based Bayes Factor and the Effective Sample Size Bayarri et. al (2019).
\end{remark}


\section{Strategies to select the calibration constant $C_\alpha$}
We now introduce some strategies for choosing the constant $C_\alpha$, which are simple enough for fast implementation in practice. Different strategies could be developed, providing alternative calibrations. 

\begin{itemize}

\item[1.]\textbf{The strategy of a simple approximation}\\
The simplest approximation in (\ref{eq1}), which is implicit in the BIC approximation, comes from assuming priors $\pi^N(\boldsymbol{\beta}_j,S_j)$, $\pi^N(\boldsymbol{\delta}_i,S_i)$ to be $N((\beta_j,\sigma_j)|(\boldsymbol{\beta}_j,S_j),\hat{I}_j^{*(-1)})$, $N((\delta_i,\sigma_i)|(\boldsymbol{\delta}_i,S_i),\hat{I}_i^{*(-1)})$ respectively, where $\hat{I}_k^*=\hat{I}_k/n^*$, with $\hat{I}_k$ being the observed Fisher Information matrix, but noting that $n^*$ is the Effective Sample Size, mentioned before. This leads to a $C=1$ in (\ref{eq5}) and then a $C_\alpha=\exp\left\{-\frac{n-j}{2(n-1)}\cdot g_{n,\alpha}(q)\right\}$ in (\ref{eq10}).

    \item[2.] \textbf{The strategy of a minimal balanced experiment: The one-way Layout}\\
We suppose that $m$ group of observations are available, with $n_k$ observations in the $k$th group, and that
$$y_{kh}\sim N(\mu_k,\sigma^2),~~~k=1,..,m~~,~~h=1,..,n_k,$$
independently, given $\mu_1,...,\mu_m,\sigma^2$. We shall denote by $M_i$ the model which sets $\mu_1=\cdots=\mu_m$, and by $M_j$ the model which allows $\mu_1\neq \cdots\neq\mu_m$. Note that $j=m$, $i=1$ and the matrices $\mathbf{X}_i$, $\mathbf{X}_j$, are easily identified. We have that $q=m-1$, thus (\ref{eq10}) is reduced to
$$\alpha(n_k,q)=\dfrac{\left[g_{n,\alpha}(q)+\log((\prod_{k=1}^{q+1}n_k)/n)\right]^{q/2-1}}{\left((\prod_{k=1}^{q+1}n_k)/n\right)^{\frac{q+1}{2(2q+1)}}\Gamma\left(\frac{q}{2}\right)}C_\alpha,$$
where $n=\displaystyle\sum_{k=1}^mn_k$. For the minimal balanced experiment, $n_k=2$, for each group and $n=2m=2(q+1)$  then, $$ C_\alpha=\alpha\cdot\dfrac{(2^q/(q+1))^{\frac{q+1}{2(2q+1)}}\Gamma\left(\frac{q}{2}\right)}{[g_{n,\alpha}(q)+\log(2^q/(q+1))]^{\frac{q}{2}-1}},$$ where $\alpha$ is the desired level for the minimal sample. The case $m=2$ is of particular interest since $q=1$, then the calibration constant $C_\alpha$ is:
$$C_\alpha=\alpha\cdot\sqrt{\pi\cdot g_{n,\alpha}(1)}.$$

\item[3.]\textbf{The strategy based in PBIC}\\
A major improvement, over BIC type approximations, has been recently introduced in (\citet{Bayarri2019}) and termed Prior Based Information Criterion (PBIC). The improvement is due to: i) "the sample size" $n$ is replaced by a much more precise "effective sample size" $n^e$ (for i.i.d. observations $n^e=n$, but not for non-i.i.d. observations), and ii) the effect of the prior is retained in the final expression, on which a flat tailed non-normal prior is employed.\\
This strategy consists in replacing in (\ref{eq5}) the constant $C$ that depends on the prior assumptions by
$$C=2\sum_{m_i=1}^{q_i}\log\frac{(1-e^{-v_{m_i}})}{\sqrt{2}v_{m_i}}-2\sum_{m_j=1}^{q_j}\log\frac{(1-e^{-v_{m_j}})}{\sqrt{2}v_{m_j}},$$
where $v_{m_l}=\frac{\hat{\xi}_{m_l}}{[d_{m_l}(1+n^e_{m_l})]}$ with $l=i,j$ corresponding to the Model $M_i$ and $M_j$ respectively, see (\citet{Bayarri2019}). Here $n^e_{m_l}$ refers to the effective sample size (called TESS, see ( \citet{Berger2014})). Hence

$$\alpha_{(b,n)}(q)=\dfrac{[g_{n,\alpha}(q)+\log(b)+C]^{\frac{q}{2}-1}}{b^{\frac{n-j}{2(n-1)}}\cdot\left(\frac{2(n-1)}{n-j}\right)^{q/2-1}\Gamma\left(\frac{q}{2}\right)}\times C_{\alpha},$$
and $$C_{\alpha}=\exp\left\{-\frac{n-j}{2(n-1)} \left( g_{n,\alpha}(q)+C \right)\right\}.$$
\end{itemize}

\section{Examples}
\subsection{Balanced One Way Anova}
	Suppose we have $k$ groups with $r$ observations each, for a total sample size of $kr$ and let 
$H_0: \mu_1= \cdots = \mu_k=\mu \;\; vs \;\; H_1: \mbox{At least one } \mu_i \mbox{ different}$. Then the design matrices for both models are:

{\footnotesize
\[ \mathbf{X}_1=\left(\begin{array}{c}
1\\ 
1\\ 
\vdots\\ 
1\\
\end{array}\right) \;, \mathbf{X}_k = \left( \begin{matrix}
1 & 0 & \ldots & 0 \\ 
1 & 0 & \ldots & 0 \\ 
\vdots & \vdots & \ldots & \vdots \\ 
1 & 0 & \ldots & 0 \\ 
0 & 1 & \ldots & 0 \\ 
0 & 1 & \ldots & 0 \\ 
\vdots & \vdots & \ldots & \vdots \\ 
0 & 1 & \ldots & 0 \\ 
\vdots & \vdots & \ldots & \vdots \\ 
0 & 0 & \ldots & 1 \\ 
0 & 0 & \ldots & 1 \\ 
\vdots & \vdots & \ldots & \vdots \\ 
0 & 0 & \ldots & 1
\end{matrix}\right) \; , b=\frac{|\mathbf{X}_k^t\mathbf{X}_k|}{|\mathbf{X}_1^t\mathbf{X}_1|}=k^{-1}r^{k-1},\]}

\noindent and the adaptive alpha for linear model in accordance with (\ref{eq10}) is
\begin{equation*}
    \alpha(k,r)=\dfrac{[g_{r,\alpha}(k-1)-\log(k)+(k-1)\log(r)]^{\frac{k-3}{2}}}{\left(k^{-1}r^{k-1}\right)^{\frac{r-1}{2(r-1/k)}}\left(\frac{2(r-1/k)}{r-1}\right)^{\frac{k-3}{2}}\Gamma\left(\frac{k-1}{2}\right)}C_{\alpha}.
\end{equation*}
Here, the number of replicas $r$ is The Effective Sample Size (TESS) see remark~\ref{obs3}.

We will use initially the strategy of selecting $C_{\alpha}$ by fixing the sample size for a designed experiment, as suggested in \citet{PP2014}, allowing us to compare our adaptive $\alpha$ for linear models with the simpler version suggested there. The experiments were designed using an effect size of $f=0.25$ ($f=\frac{\mu_1-\mu_2}{\sigma}$), which according to \citet{Cohen1988} represents a medium effect size. We fixed  $\alpha=0.05$ and the power at $0.8$ . The sample sizes obtained were $r_0 = 64, 40 \mbox{ and } 26$ for $k=2, 5 \mbox{ and }10$, respectively. The results shown Table \ref{tab1} evidence that both corrections for $\alpha$ yield very similar results, with the significance level decreasing steadily with the number of replicates.
	
\begin{table}[h]
     \centering
	{\small  \begin{tabular}{|r|rrr|rrr|}
			\hline 
			& \multicolumn{6}{c|}{$k$}  \\ 
			\hline 
			&\multicolumn{3}{c|}{Adaptive $\alpha$ for linear model}  &\multicolumn{3}{c|}{Adaptive $\alpha$ (PP 2014)} \\ 
			\hline 
			r & \multicolumn{1}{c}{2} & \multicolumn{1}{c}{5 }& \multicolumn{1}{c|}{10} & \multicolumn{1}{c}{2 }& \multicolumn{1}{c}{5} & \multicolumn{1}{c|}{10}  \\ 
			\hline 
			$50$ &$0.057$  &$0.0327$  &$3.6\times 10^{-3} $ &$ 0.058$ &  $0.0333$  &  $3.8 \times 10^{-3}$ \\ 
			$100$ &$0.038$& $0.0087$ & $2.2\times 10^{-4} $&$0.038$  & $0.0093$ & $2.4\times 10^{-4}$  \\  
			$500$ & $0.016$ & $0.0004$ &$3.1 \times 10^{-7}$  & $0.015$ & $0.0005$ & $3.4\times 10^{-7}$  \\ 
			$1000$ & $0.011$& $0.0001$ & $1.8\times 10^{-8}$ & $0.010$ &$0.0001$  & $2.0 \times 10^{-8}$  \\ 
			\hline 
	\end{tabular} }
\caption{Adaptive $\alpha$ for linear models vs. \citet{PP2014} original adaptive $\alpha$ for the One Way Anova problem. In both cases, a calibration strategy based on a designed experiment is used. }
      \label{tab1}
    \end{table}
  
\begin{table}[h]
    \centering
   {\small  \begin{tabular}{|c|c|c|c|}
	\hline 
	& \multicolumn{3}{c|}{$k$}  \\ 
	\hline 
	 &\multicolumn{1}{c|}{Minimal sample}  &\multicolumn{1}{c|}{Simple Calibration}&\multicolumn{1}{c|}{PBIC Calibration} \\ 
	\hline 
	r & \multicolumn{1}{c|}{2} & \multicolumn{1}{c|}{2}& \multicolumn{1}{c|}{2} \\ 
	\hline 
	$4$ & $0.0523$ & $0.0360$ & $0.0283$ \\ 
	$10$ & $0.0342$ & $0.0235$ & $0.0159$ \\ 
	$50$ &$0.0130$ & $ 0.0090$ & $0.0061$\\ 
	$100$ &$0.0087$& $0.0060$  & $0.0041$ \\  
	$500$ & $0.0035$ &$0.0024$ & $0.0017$ \\ 
	$1000$ & $0.0024$& $0.0017$& $0.0011$  \\ 
	\hline 
\end{tabular} }
\caption{Adaptive alpha for linear models in the One Way Anova setting for the three calibration strategies described in Section 4.}
    \label{tab2}
\end{table}

Table~\ref{tab2} shows how the three different calibration strategies discussed in Section 4 decrease the adaptive alpha as the effective sample size grows. It is reassuring that the different strategies yield comparable results, with the strategy based on PBIC being somewhat more drastic (in this case but not in general, see sect.7.2) in its penalization for higher samples.

We now present a simulation that shows how our methodology for decreasing alpha works precisely to improve scientific inference and interpretation. Inspired by an example in \citet{SBB2001} we perform the following experiment:  we simulate $r$ data points from each of two normal distributions $N(\mu_1, \sigma)$ and $N(\mu_2,\sigma)$. We replicate this $2K$ times. For $K$ of the simulations, $\mu_1-\mu_2=0$, while for the other $K$ $\mu_1-\mu2=f>0$.  For all $2K$ replications, we test the hypotheses $H_0: \mu_1=\mu2$ vs $H_1: \mu_1 \neq \mu_2$, and then count  how many of the p-values lie between $0.05 - \varepsilon$ and $0.05$. Note that all these p-values would be deemed enough for rejecting $H_0$ if $\alpha=0.05$ is selected.  Finally, we determine the proportion of these "significant" p-values obtained from samples where $H_0$ is true. 

Table 3 presents the median percentage of these significant p-values coming from samples where $H_0$ is true for 100 replicates of the simulation scheme with $K=4000$, $f=0.25$, $\sigma=1$ and $\varepsilon=0.04$, for $r=10, 50, 100, 500$ and $1000$.  The usual (flawed) interpretation is that only about $5\%$ should be generated from $H_0$. However, the reality is that much more are false positives. The proportion is not monotonic with $r$, but always far higher than $5\%$. In the limit, for $r=1000$, almost 100\% of these significant values near $0.05$ are generated from $H_0$. This result should not be surprising, as for large sample sizes and fixed significance levels, most null hypotheses are rejected. Table 3 also presents the proportion of significant p-values coming from $H_0$ when the $\alpha$ level is corrected according to the method suggested in this paper, using a calibration strategy based on PBIC. With this correction, the proportion of false positives decreases steadily, providing a more reliable  Type I error  control. 

\begin{table}[h]
    \centering
    {\small  \begin{tabular}{|c|c|c|}
	\hline 
	& \multicolumn{2}{c|}{\% of samples with $0.01<p<0.05$}  \\ 
	\hline 
	 &\multicolumn{1}{c|}{without adjustment}  &\multicolumn{1}{c|}{with PBIC calibration } \\ 
	\hline 
	r & \multicolumn{1}{c|}{2-groups} & \multicolumn{1}{c|}{2-group}\\ 
	\hline

	$10$ & $39.06\%$ & $34.18\%$ \\ 
	$50$ &$21.43\%$ & $8.57\%$ \\ 
	$100$ &$15.73\%$& $3.07\%$  \\  
	$500$ & $ 39.04\%$ &$0.22\%$ \\ 
	$1000$ & $97.15\%$& $0.11\%$\\ 
	\hline 
\end{tabular} }
\caption{Percentage of p-values between $0.01$ and $0.05$ considered significant coming from data generated under the null hypothesis when equal number of tests with $H_0$ true and $H_0$ false are generated for different sample sizes. Uncorrected and corrected alpha levels are considered.}
\label{tab3}
\end{table}

\subsection{Linear Regression Model}

Consider a linear regression model $M_j: y_v=\beta_1+\beta_2x_{v2}+\cdots+\beta_j x_{vj}+\epsilon_v$ 
with $1\leq v\leq n$ and $2\leq j\leq k$, then $$|\mathbf{X}_j^t\mathbf{X}_j|=n(n-1)^{j-1}\prod_{l=2}^{j}s_{l}^2|R_j|$$ where $s_l^2$ and $R_j$ is the variance  and the correlation matrix of the predictors in model $M_j$ respectively, so in the adaptive alpha in (\ref{eq10})
\begin{equation}\label{eq11}
    b=(n-1)^{j-i}\left(\prod_{l=i+1}^{j}s^2_{l}\right)|R_{j-i}-R_{ij}^tR_i^{-1}R_{ij}|,
\end{equation}

\noindent Here $R_{ij}$ is the correlation matrix between predictors of the models $M_j$ that are not in $M_i$ with predictors of the model $M_i$, and $R_{j-i}$ is the correlation matrix of the predictors of the models $M_j$ that are not in $M_i$, see Appendix 2 for more detail.

As an example, we analyze a data set taken from \citet{acuna} which can be accessed at \url{http://academic.uprm.edu/eacuna/datos.html}. We want to predict the average mileage per gallon (denoted by \texttt{mpg}) of a set of $n=82$ vehicles using four possible predictor variables:  cabin capacity in cubic feet (\texttt{vol}), engine power (\texttt{hp}), maximum speed in miles per hour (\texttt{sp}) and vehicle weight in hundreds of pounds (\texttt{wt})



To study the effect of the variances of the predictors and their correlations on the proposed adaptive alpha, we will compare the following models,

\begin{enumerate}
    \item[1.] $H_0:M_2:$(mpg=$\beta_1$+$\beta_2\text{wt}_i$+$\epsilon_i$) vs $H_1:M_3:$(mpg=$\beta_1$+$\beta_2\text{wt}_i$+$\beta_3\text{sp}_i$+$\epsilon_i$)
    \item[2.] $H_0:M_2:$(mpg=$\beta_1$+$\beta_2\text{wt}_i$+$\epsilon_i$) vs  $H_1:M_3:$(mpg=$\beta_1$+$\beta_2\text{wt}_i$+$\beta_3\text{hp}_i$+$\epsilon_i$)
    \item[3.] $H_0:M_2:$(mpg=$\beta_1$+$\beta_2\text{wt}_i$+$\epsilon_i$) vs $H_1:M_3:$(mpg=$\beta_1$+$\beta_2\text{wt}_i$+$\beta_3\text{vol}_i$+$\epsilon_i$)
\end{enumerate}

For all these tests (\ref{eq11}) can be rewritten as
$$b=(n-1)s_{3}^2(1-\rho_{23}^2),$$
where $s_{3}^2$ is the variance of the entering predictor in model $M_3$ and $\rho_{23}$ is the correlation between $\text{wt}$ y the new predictor in $M_3$. 

Table  \ref{tab4} shows the correction of the significance $\alpha$ (using $0.05$ as initial value) through the proposed adaptive alpha using simple and PBIC calibrations 

\begin{table}
    \centering
\begin{tabular}{ccrrrrrr}
\hline
Test &\small{Predictor}&\multicolumn{1}{c}{\small Var($\cdot$)} & \small{Cor(wt,$\cdot$)} & \multicolumn{1}{c}{b} & p-value  & \small{$\alpha_{Simple}$} & \small{$\alpha_{PBIC}$}\\ 
\hline
1 & \texttt{sp} &197.1 &0.68 &8612.9& 0.0325 & 0.0004  & 0.0134\\
2 & \texttt{hp} & 3230.9 & 0.83 & 80449.5& 0.1661 & 0.0001 & 0.0046\\
3 & \texttt{vol} & 491.3 &0.38 &  33901.1& 0.6482 & 0.0002  & 0.0087\\ 
\hline
\end{tabular} 
    \caption{ Effect of the variances and correlations of the predictors on the calibration of significance level $\alpha$ in the modelling of the efficiency (mpg) for a set of cars \citep{acuna} when the simple calibration and the PBIC calibration are used.}
    \label{tab4}
\end{table}

In all cases, the significance level is substantially reduced, specially using the simple callibration. Note that the strongest correction (the largest reduction on the $\alpha$ level) corresponds to the engine power (\texttt{hp}). This variable has both the largest variance and the highest correlation with the weight. For test 1 (including speed \texttt{sp} in the model), the use of the adaptive alpha levels will change the inference, as $p=0.0325$ is smaller than $0.05$ but larger than the calibrated $\alpha$ values. In the other two cases, the inference is not altered. 

\vspace{0.25cm}

  
 
 \subsection{Comparing significance level calibrations based on BIC and PBIC}
 
We will now present two challenging examples, and compare the behavior of the Adaptive Alpha Aignificance Levels based on BIC \citep[as presented in][]{PP2014} and the Adaptive Alpha for Linear Models calibrated using PBIC proposed in this paper.
 
 \subsubsection{Comparing Findley's counterexample: adaptive alpha via (BIC vs. PBIC)}
 Consider the following simple linear model
 
 $$Y_i=\frac{1}{\sqrt{i}}\cdot\theta+\epsilon_i, ~~\text{where}~\epsilon_i\sim N(0,1),\\
 i=1,2,3,..,n$$
 and we are comparing the models $H_0:\theta=0$ and $H_1:\theta\neq 0$. This  is a Classical and challenging counter example against BIC and the Principle of Parsimony. In \citet{Bayarri2019} it is shown the inconsistency of BIC but the consistency of PBIC in this problem. 
 
 Here we will compare two calibration strategies for the significance level $\alpha$:
 
 \begin{enumerate}
\item Adaptive Alpha via BIC \citep{PP2014}:
 $$\alpha_n(q)=\frac{[\chi_\alpha^2(q)+q\log(n)]^{\frac{q}{2}-1}}{2^{\frac{q}{2}-1}n^{\frac{q}{2}}\Gamma\left(\frac{q}{2}\right)}\times C_\alpha$$
 $C_\alpha$ based on a simple calibration.

\item Adaptive Alpha for Linear Models calibrated using PBIC:
 $$\alpha_{(b,n)}(q)=\dfrac{[g_{n,\alpha}(q)+\log(b)+C]^{\frac{q}{2}-1}}{b^{\frac{n-j}{2(n-1)}}\cdot\left(\frac{2(n-1)}{n-j}\right)^{q/2-1}\Gamma\left(\frac{q}{2}\right)}\times C_{\alpha},$$
$$C_{\alpha}=\exp\left\{-\frac{n-j}{2(n-1)} \left( g_{n,\alpha}(q)+C \right)\right\}.$$
$$C=-2\log\frac{(1-e^{-v})}{\sqrt{2}v}, v=\frac{\hat{\theta}^2}{d(1+n^e)}, d=\left(\sum_{i=1}^{n}\frac{1}{i}\right)^{-1}, n^e=\sum_{i=1}^{n}\frac{1}{i}$$ 
\end{enumerate}

Table~\ref{tab5} shows the behavior for both strategies when $n$ grows and we consider $\alpha=.05$ and $q=1$. The adaptive alpha for PBIC corrects far more slowly, as it should, due to the decreasing information content of successive observations.

\begin{table}[h]
    \centering
    {\small  \begin{tabular}{|c|c|c|}
	\hline 
	n &\multicolumn{1}{c|}{Adaptive Alpha via PBIC} &\multicolumn{1}{c|}{Adaptive Alpha via BIC}  \\ 
	\hline 
	$10$ & 0.0457 &0.0149\\ 
	$20$ & 0.0367 &0.0100\\ 
	$30$ & 0.0338 &0.0079 \\  
	$40$ & 0.0319 &0.0070 \\ 
	$50$ & 0.0307 & 0.0059\\ 
	$100$ & 0.0282 & 0.0040\\
	$1000$ & 0.0253 & 0.0011\\
	$10000$ & 0.0248 & 0.0003\\
	\hline 
\end{tabular} }
\caption{Behavior of the Adaptive Alpha based on BIC (simple aproximation) vs the Adaptive Alpha for Linear Models calibrated using PBIC for Findley's counterexample when sample size $n$ increases}
\label{tab5}
\end{table}

\subsubsection{Adaptive alpha via PBIC for testing equality of two means with unequal variances}
Consider comparing two normal means via the test $H_0:\mu_1=\mu_2$ versus $H_1:\mu_1\neq\mu_2$, where the associated known variances, $\sigma^2_1$ and $\sigma^2_2$ are not equal. $$\mathbf{Y}=\mathbf{X}\mathbb{\mu}+\mathbb{\epsilon}=\begin{pmatrix}
1&0\\
\vdots&\vdots\\
1& 0\\
0& 1\\
\vdots&\vdots\\
0&1
\end{pmatrix}\begin{pmatrix}
\mu_1\\
\mu_2
\end{pmatrix}+\begin{pmatrix}
\epsilon_{11}\\
\vdots\\
\epsilon_{2n_2}
\end{pmatrix},$$
$$\times \mathcal{\epsilon}\sim N(\mathbf{0},\text{diag}\{\underbrace{\sigma_1^2,...,\sigma_1^2}_{n_1},\underbrace{\sigma_2^2,...,\sigma_2^2\}}_{n_2})$$
Defining $\alpha=(\mu_1+\mu_2)/2$ and $\beta=(\mu_1-\mu_2)/2$ places this in the linear model comparison framework,
$$\mathbf{Y}=\mathbf{B}\binom{\mathbb{\alpha}}{\mathbb{\beta}}+\mathbb{\epsilon}$$
with $$\mathbf{B}=\begin{pmatrix}
1&\frac{1}{2}\\
\vdots&\vdots\\
1& \frac{1}{2}\\
1& -\frac{1}{2}\\
\vdots&\vdots\\
1&-\frac{1}{2}
\end{pmatrix}$$
where we are comparing $M_0:\beta=0$ versus $M_1:\beta\neq 0$.\\
So for adaptive alpha via PBIC,
$$C=-2\log\frac{(1-e^{-v})}{\sqrt{2}v}$$
$v=\frac{\hat{\beta}^2}{d(1+n^e)}, d=\left(\frac{\sigma_1^2}{n_1}+\frac{\sigma_2^2}{n_2}\right), n^e=\max\left\{\frac{n_1^2}{\sigma_1^2},\frac{n_2^2}{\sigma_2^2}\right\}\left(\frac{\sigma_1^2}{n_1}+\frac{\sigma_2^2}{n_2}\right)$ and the adaptive alpha behaves according to the table~\ref{tab6}, here $\alpha=.05$. 

Table~\ref{tab6} shows how the adaptive alpha based on PBIC adapts faster when the sample size for the group with smaller variance increases. On the contrary, the adaptive alpha based on BIC calibration adapts at a similar pace for both groups.      
\begin{table}[h]
    \centering
    {\small  \begin{tabular}{|c|c|c|c|}
	\hline 
	&& Adaptive alpha via PBIC& Adaptive alpha via BIC  \\ 
	\hline 
	$n_1$ & $n_2$ & \multicolumn{1}{c|}{$\sigma^2_1=14$, $\sigma^2_2=140$}& $n=n_1+n_2$\\ 
	\hline

	10 & $10$ &0.0159 &  0.0099 \\ 
	10 & $100$ &0.0111 & 0.0038 \\ 
	10 & $500$ &0.0104 &  0.0016 \\
	100 & $10$ &0.0110 &  0.0038\\
	100 & $100$ &0.0042 & 0.0027\\ 
	100 & $500$ &0.0030 & 0.0015 \\ 

	\hline 
\end{tabular} }
\caption{Adaptive alpha via PBIC and BIC for testing equality of two means when variances are unequal. }
\label{tab6}
\end{table}

\section{Final comments, questions and some answers}
\begin{itemize}
   \item[1.] The adaptive $\alpha$ provides guidance for adjusting significance to the sample size. The Linear Model version incorporates not only the sample size and the difference of dimensions, but also the information provided by the predictors or the design, and particularly their correlations, correcting for co-linearity.
\item[2.] The adaptive $\alpha$ is simple to use, and gives equivalent results than a sensible Bayes Factor,like Bayes Factors with Intrinsic Priors,  but easier to be employed by practitioners, even by those who are not trained in sophisticated Bayesian Statistics. We hope that this development will give tools to the practice of Statistics.
\item[3.] The results exposed here, make use of state of the art large sample approximations of Bayes Factors like the PBIC and can be coupled with recent sensible base thresholds like $\alpha=0.005$, \citet{Benjamin2017}. 

\end{itemize}

\section*{Acknowledgements}
The work of M.E. P\'erez and L.R. Pericchi has been partially funded by NIH grants U54CA096300, P20GM103475 and R25MD010399. 

\bibliographystyle{chicago}
\bibliography{ARLB}
\section*{Appendix 1 The likelihood ratio}
Define $$r(\mathbf{y}|(\mathbf{X}_i,\mathbf{X}_j))=\dfrac{f(\mathbf{y}|\mathbf{X}_i\widehat{\boldsymbol{\delta}}_i,S_i^2\mathbf{I}_n)}{f(\mathbf{y}|\mathbf{X}_j\widehat{\boldsymbol{\beta}}_j,S_j^2\mathbf{I}_n)}$$
we will perform the calculations for the hypothesis test $$H_0:\text{Model}\hspace{.1cm} M_i\hspace{0.3cm} versus \hspace{.3cm}H_1:\text{Model}\hspace{.1cm} M_j.$$

Indeed, for model $M_i$ $$L(\mathbf{y}|\mathbf{X}_i,\sigma_i^2,\boldsymbol{\delta}_i)=\frac{1}{(2\pi)^{n/2}(\sigma_i^2)^{n/2}}\exp\left\{-\frac{1}{2\sigma_i^2}(\mathbf{y}-\mathbf{X}_i\boldsymbol{\delta}_i)^t(\mathbf{y}-\mathbf{X}_i\boldsymbol{\delta}_i)\right\}.$$

Since the MLE of $\boldsymbol{\delta}_i$ is $\widehat{\boldsymbol{\delta}}_i=(\mathbf{X}_i^t\mathbf{X}_i)^{-1}\mathbf{X}_i^t\mathbf{y}$ and the MLE of $\sigma_i^2$ is $S_{i}^2=\dfrac{\mathbf{y}^t(\mathbf{I}-\mathbf{H}_i)\mathbf{y}}{n}$, where $\mathbf{H}_i=\mathbf{X}_i(\mathbf{X}_i^t\mathbf{X}_i)^{-1}\mathbf{X}_i^t$

$$\sup_{\Omega_0}L(\mathbf{y}|\mathbf{X}_i,\sigma_i^2,\boldsymbol{\delta}_i)=\frac{1}{(2\pi)^{n/2}(S_{i}^2)^{n/2}}\exp\left\{-\frac{n}{2}\right\}.$$

For model $M_j$ $$L(\mathbf{y}|\mathbf{X}_j,\sigma_j^2,\boldsymbol{\beta}_j)=\frac{1}{(2\pi)^{n/2}(\sigma_j^2)^{n/2}}\exp\left\{-\frac{1}{2\sigma_j^2}(\mathbf{y}-\mathbf{X}_j\boldsymbol{\beta}_j)^t(\mathbf{y}-\mathbf{X}_j\boldsymbol{\beta}_j)\right\}.$$

Since MLE of $\boldsymbol{\beta}_j$ is $\widehat{\boldsymbol{\beta}}_j=(\mathbf{X}_j^t\mathbf{X}_j)^{-1}\mathbf{X}_j^t\mathbf{y}$ and the MLE of $\sigma_j^2$ is $S_{j}^2=\dfrac{\mathbf{y}^t(\mathbf{I}-\mathbf{H}_j)\mathbf{y}}{n}$

$$\sup_{\Omega}L(\mathbf{y}|\mathbf{X}_j,\sigma_j^2,\boldsymbol{\beta}_j)=\frac{1}{(2\pi)^{n/2}(S_{j}^2)^{n/2}}\exp\left\{-\frac{n}{2}\right\}.$$

Thus the likelihood ratio is $$r(\mathbf{y}|(\mathbf{X}_i,\mathbf{X}_j))=\dfrac{\sup_{\Omega_0}L(\mathbf{y}|\mathbf{X}_i,\sigma_i^2,\boldsymbol{\alpha}_i)}{\sup_{\Omega}L(\mathbf{y}|\mathbf{X}_j,\sigma_j^2,\boldsymbol{\beta}_j)}=\left(\frac{S_{j}^2}{S_{i}^2}\right)^{\frac{n}{2}}=\left(\dfrac{\mathbf{y}^t(\mathbf{I}-\mathbf{H}_j)\mathbf{y}}{\mathbf{y}^t(\mathbf{I}-\mathbf{H}_i)\mathbf{y}}\right)^{\frac{n}{2}}.$$ 

\section*{Appendix 2 An expression for b in (\ref{eq10})}

Consider linear regression model $M_j: y_v=\beta_1+\beta_2x_{v2}+\cdots+\beta_j x_{vj}+\epsilon_v$ 
with $1\leq v\leq n$ and $2\leq j\leq k$, then

$$\mathbf{X}_j=\begin{bmatrix}
1& x_{12}-\bar{x}_2&\cdots&x_{1j}-\bar{x}_j\\
1& x_{22}-\bar{x}_2&\cdots&x_{2j}-\bar{x}_j\\
\vdots&\vdots&\vdots&\vdots\\
1&x_{n2}-\bar{x}_2&\cdots&x_{nj}-\bar{x}_j\\
\end{bmatrix}$$
and
$$\mathbf{X}_j^t\mathbf{X}_j=\begin{bmatrix}
n& 0&0&\cdots&0\\
0& (n-1)s_2^2&(n-1)s_2s_3\rho_{23}&\cdots&(n-1)s_2s_j\rho_{2j}\\
\vdots&\vdots&\vdots&\vdots&\vdots\\
0&(n-1)s_2s_j\rho_{2j}&(n-1)s_3s_j\rho_{2j}&\cdots&(n-1)s_j^2\\
\end{bmatrix}$$
then 
$$|\mathbf{X}_j^t\mathbf{X}_j|=n(n-1)^{j-1}\begin{vmatrix}
s_2^2& s_2s_3\rho_{23}&\cdots&s_2s_j\rho_{2j}\\
s_2s_3\rho_{23}& s_3^2&\cdots&s_3s_j\rho_{3j}\\
\vdots&\vdots&\vdots&\vdots\\
s_2s_j\rho_{2j}&s_3s_j\rho_{3j}&\cdots&s_j^2\\
\end{vmatrix},$$
note that row $l$ and column $l$ are multiplied by $s_l$, using properties of the determinants 

$$|\mathbf{X}_j^t\mathbf{X}_j|=n(n-1)^{j-1}s_2^2s_3^2\cdots s_j^2\begin{vmatrix}
1& \rho_{23}&\cdots&\rho_{2j}\\
\rho_{23}& 1&\cdots&\rho_{3j}\\
\vdots&\vdots&\vdots&\vdots\\
\rho_{2j}&\rho_{3j}&\cdots&1\\
\end{vmatrix}=n(n-1)^{j-1}\prod_{l=2}^{j}s_{l}^2|R_j|$$
on the other hand, 

$$R_j=\begin{bmatrix}
1& \rho_{23}&\cdots&\rho_{2j}\\
\rho_{23}& 1&\cdots&\rho_{3j}\\
\vdots&\vdots&\vdots&\vdots\\
\rho_{2j}&\rho_{3j}&\cdots&1\\
\end{bmatrix}=\begin{bmatrix}
R_i& R_{ij}\\
R_{ij}&R_{j-i}\\
\end{bmatrix}$$
where 
$$R_{ij}=\begin{bmatrix}
\rho_{2j+1}& \rho_{3j+1}&\cdots&\rho_{ii+1}\\
\rho_{2j+2}& \rho_{3j+2}&\cdots&\rho_{ii+2}\\
\vdots&\vdots&\vdots&\vdots\\
\rho_{2j}&\rho_{3j+2}&\cdots&\rho_{ij}\\
\end{bmatrix}~~\text{and}~~R_{j-i}=\begin{bmatrix}
1& \rho_{i+2i+1}&\cdots&\rho_{ji+1}\\
\rho_{i+1i+2}& 1&\cdots&\rho_{ji+2}\\
\vdots&\vdots&\vdots&\vdots\\
\rho_{i+1j}&\rho_{i+2j}&\cdots&1\\
\end{bmatrix}.$$
Now since $\mathbf{X}_j$ is a full rank matrix, it can be seen that $$|R_j|=|R_i||R_{j-i}-R_{ij}^tR_i^{-1}R_{ij}|$$ 
thus $$ b=\frac{|\mathbf{X}_j^t\mathbf{X}_j|}{|\mathbf{X}_i^t\mathbf{X}_i|}=(n-1)^{j-i}\left(\prod_{l=i+1}^{j}s^2_{l}\right)|R_{j-i}-R_{ij}^tR_i^{-1}R_{ij}|$$
\end{document}